\documentclass[12pt]{article}
\usepackage{amsfonts,amssymb,amsmath}

\numberwithin{equation}{section}
\newtheorem{theorem}{Theorem}[section]
\newtheorem{lemma}[theorem]{Lemma}
\newtheorem{proposition}[theorem]{Proposition}

\newtheorem{defin}[theorem]{Definition}

\newenvironment{proof}{\noindent \textbf{Proof: }}{\hfill
$\Box$  \vspace{1ex}}
\newenvironment{definition}{\begin{defin}\em}{\end{defin}}
\newtheorem{defins}[theorem]{Definitions}

\newtheorem{exs}[theorem]{Examples}

\newtheorem{ex}[theorem]{Example}

\newtheorem{rem}[theorem]{Remark}

\newtheorem{rems}[theorem]{Remarks}

\def\fq{\mathbb{F}_q}

\def\N_0{\mathbb{N}_0}

\def\RM{\textrm{RM}}
\def\prm{\textrm{PRM}}

\begin{document}


\begin{center}
{\Large\textbf{The next-to-minimal weights of binary projective Reed-Muller 
codes}}
\end{center}
\vspace{3ex}

\noindent
\begin{center} 
\textsc{C\'{\i}cero  Carvalho and V. G. Lopez Neumann}\footnote{\noindent 
Authors' emails: cicero@ufu.br and victor.neumann@ufu.br. Both authors were 
partially supported by  grants from CNPq and FAPEMIG.\\
Published in IEEE Transactions on Information Theory, vol.\ 62, issue
11, Nov.\ 2016. \\  http://doi.org/10.1109/TIT.2016.2611527} 
\end{center}

\noindent
{\small Faculdade de Matem\'atica,  Universidade Federal de Uberl\^andia,    Av.\ J.\ N.\ \'Avila 2121, 38.408-902 \ \ Uberl\^andia - MG, Brazil}
\vspace{3ex}

\noindent
{\footnotesize \textbf{Abstract.} 
Projective Reed-Muller codes were introduced by Lachaud, in 1988 and their 
dimension and minimum distance were determined by Serre and S\o{}rensen in 
1991. In 
coding 
theory one is also interested in the higher Hamming weights, to study the code 
performance. Yet, not many values of the higher Hamming weights are known for 
these codes, not even the second lowest weight (also known as next-to-minimal 
weight) is completely determined. In this paper we determine all the values of 
the next-to-minimal weight for the binary projective Reed-Muller codes, which 
we show to be equal to the next-to-minimal weight of Reed-Muller 
codes in most, but not all, cases.}
\vspace{2ex}

\section{Introduction} \label{intro}

The so-called Reed-Muller codes appeared in 1954, when they were defined by 
D.E.\ Muller (\cite{muller}) and given a decoding algorithm by I.S.\ Reed 
(\cite{reed}). They were codes defined over $\mathbb{F}_2$ and in 1968  Kasami, 
Lin, and Peterson (\cite{kasami}) extended the original definition to a finite 
field $\mathbb{F}_q$, where $q$ is any prime power. They named these codes 
``generalized Reed-Muller codes'' and presented some results on the weight 
distribution, the dimension of the codes being determined in later works. In 
coding theory one  is always interested in the values of the higher Hamming 
weights of a code because of their relationship with the code performance, but 
usually this is not a simple problem. For the generalized Reed-Muller codes, 
the complete determination of the second lowest Hamming weight, also called 
next-to-minimal weight,  
was only completed in 2010, when Bruen (\cite{bruen}) observed that the value 
of these weights could be obtained from unpublished results in the Ph.D.\ 
thesis 
of D.\ Erickson (\cite{erickson}) and Bruen's own results from 1992 and 2006. 

Twenty years after the definition of the generalized Reed-Muller codes the 
class of projective Reed-Muller codes  was introduced by Lachaud 
(\cite{lachaud}). The parameters of these codes were determined by Serre 
(\cite{serre}), for some cases, and by S\o{}rensen (\cite{sorensen}) for the 
general case, and they proved that the minimum distance of the projective 
Reed-Muller codes of order $d$ is equal to the minimum distance of the 
generalized Reed-Muller code of order $d - 1$ (see  \eqref{serre-sor}). The 
determination of the next-to-minimal weight for these codes 
is yet to be done, and  there are some results (also about higher Hamming 
weights) on this subject by Rodier and Sboui (\cite{rodier-sboui}, 
\cite{rodier-sboui2}, \cite{sboui1}) and also by Ballet and Rolland 
(\cite{ballet}). In this paper we present all the values  of the 
next-to-minimal weights for the case of binary projective 
Reed-Muller codes. Interestingly, we note that the next-to-minimal weight of 
the binary projective Reed-Muller codes of order $d$ is equal to the 
next-to-minimal weight of the Reed-Muller codes of order $d - 1$ in most but 
not 
all cases (see Theorem \ref{main}).
In the next section we recall the definitions of the 
generalized and projective Reed-Muller codes, and prove some results of 
geometrical nature that will allow us to determine the next-to-minimal weight 
of the binary projective Reed-Muller codes, which is done in the last 
section. 

\section{Preliminary results}

Let $\fq$ be a finite field and let $I_q = (X_1^q - X_1, \ldots, X_n^q - X_n) 
\subset \fq[X_1, \ldots, X_n]$ be the ideal of polynomials which vanish at all 
points $P_1, \ldots, P_{q^n}$ of the affine space $\mathbb{A}^n(\fq)$. Let 
$\varphi: \fq[X_1, \ldots, X_n]/I_q \rightarrow \fq^{q^n}$ be the $\fq$-linear 
transformation given by $\varphi(g + I_q) = (g(P_1) , \ldots, g(P_{q^n}) )$.

\begin{definition}
Let $d$ be a nonnegative integer.  The generalized Reed-Muller code of order 
$d$ is defined 
as $\RM(n, d) = \{ \varphi(g + I_q) \, | \, g = 0 \textrm{ or } \deg(g) \leq d 
\}$.  
\end{definition}

One may show that $\RM(n, d) = \fq^{q^n}$ if $d \geq n(q - 1)$, so in this case 
the minimum distance is 1. Let $d \leq n(q -1)$ and write $d = a (q - 1) + 
b$ with $0 < b \leq q - 1$, then the minimum distance of $\RM(n, d)$ is 
\[
W^{(1)}_{\RM}(n, d) = (q - b) q^{n - a - 1}.
\] 
The next-to-minimal 
 weight  $W^{(2)}_{\RM}(n, d)$ of 
$\RM(n, d)$ is equal to 
\[
W^{(2)}_{\RM}(n,d) =  (q - b)q^{n-a-1} + c q^{n-a-2}   = \left(1 + 
\dfrac{c}{(q- b)q}  \right) (q - b)q^{n-a-1} 
\]
where $c$ is equal to  $b -1$, $q-1$ or $q$, according to the values of $q$ 
and $d$ (see \cite[Theorems 9 and 10]{ballet}). We will quote specific values 
of $W^{(2)}_{\RM}(n,d)$ when we need them.	

Let $Q_1,\ldots, Q_N$ be the points of $\mathbb{P}^n(\fq)$, where $N = q^n + 
\ldots + q + 1$. It is known (see e.g. \cite{mercier-rolland} or 
\cite{idealapp}) that the homogeneous ideal $J_q \subset \fq[X_0, \ldots 
, X_n]$ of the polynomials which vanish in all points of $\mathbb{P}^n(\fq)$ is 
generated by $\{ X_j^q X_i - X_i^q X_j\, | 
\, 0 \leq i < j \leq n\}$. We denote by $\fq[X_0, \ldots , X_n]_d$
 (respectively, $(J_q)_d$) the $\fq$-vector subspace formed by the 
homogeneous polynomials of degree $d$ (together with the zero polynomial) in  
$\fq[X_0, \ldots , X_n]$ (respectively, $J_q$).

\begin{definition}
Let $d$ be a positive integer and 
let $\psi: \fq[X_0, \ldots , X_n]_d / (J_q)_d\rightarrow \fq^N$ be the 
$\fq$-linear transformation given by $\psi( f + (J_q)_d ) = (f(Q_1) \ldots, 
f(Q_n) )$, where we write the points of $\mathbb{P}^n(\fq)$ in the standard 
notation, i.e.\ the first nonzero entry from the left is equal to 1. The 
projective Reed-Muller code of order $d$, denoted by $\prm(n, d)$, is the image 
of 
$\psi$.
\end{definition}
The minimum distance $W^{(1)}_{\prm}(n, d)$ of $\prm(n, d)$ was determined by 
Serre  (\cite{serre}) and 
S\o{}rensen (see \cite{sorensen}) who proved that 
\begin{equation} \label{serre-sor}
W^{(1)}_{\prm}(n, d) = 
W^{(1)}_{\RM}(n, d-1).
\end{equation}
Let $\omega$ be the Hamming weight of $\varphi(g + I_q)$, where $g \in \fq[X_1, 
\ldots, X_n]$ is a polynomial of degree $d - 1$, and let $g^{(h)}$ be the 
homogenization of $g$ with respect to $X_0$. Then the degree of $g^{(h)}$ is $d 
- 1$ 
and the weight of $\psi(X_0 g^{(h)} + (J_q)_d)$ is $\omega$. This shows that, 
denoting by $W^{(2)}_{\prm}(n, d)$ the next-to-minimal weight of $\prm(n, d)$,  
we 
have 
\begin{equation}\label{2.1}
W^{(2)}_{\prm}(n, d) \leq W^{(2)}_{\RM}(n, d - 1).
\end{equation}
In the next section 
we 
will 
prove that, for binary projective Reed-Muller 
codes, equality 
holds in most but not all cases (see Theorem \ref{main}).

\begin{definition}
Let $f \in \fq[X_0, \ldots , X_n]_d$. The set of points of $\mathbb{P}^n(\fq)$ 
which are not zeros of $f$ is called the support of $f$, and we denote its 
cardinality by $| f |$ (hence $|f|$ is the weight of the codeword $\psi(f + 
(J_q)_d)$).
\end{definition}

\begin{lemma} \label{suporte}
Let $f \in \fq[X_0,\ldots , X_n]_d$ be a nonzero polynomial, and let $S$ be it 
support.  
Let  $G \subset \mathbb{P}^n(\fq)$ be a linear subspace
of  dimension $r$, with $r \in \{1, \ldots, n-1\}$, then  either $S \cap G = 
\emptyset$ or
$|S \cap G| \ge 
W^{(1)}_{\prm}(r, d)$.
\end{lemma}
\begin{proof}
After a projective transformation we may assume that $G$ is given by 
$X_{r+1}=\cdots = X_n =0$. Assume that $|S \cap G| \neq \emptyset$ and let $g$ 
be the polynomial obtained from $f$ by 
evaluating $X_i = 0$ for $i = r + 1, \ldots, n$. Then $g$ is a nonzero 
homogeneous 
polynomial of degree $d$ and its support is equal to $S \cap G$. Considering 
$g$ as a polynomial which evaluates at points of $\mathbb{P}^r(\fq)$ we have $| 
S \cap G | = | 
g | \geq W^{(1)}_{\prm}(r, d)$.
\end{proof}

Observe that when $d = 1$ we do not have a next-to-minimal weight for 
$\prm(n, 1)$ 
since all hyperplanes in $\mathbb{P}^n(\fq)$ have the same number of zeros. In
\cite[Remark 3]{sorensen} S\o{}rensen proved that  
$\prm(n,d) = \mathbb{F}_q^N$ whenever $d \geq n (q - 1) + 1$, so  
from now on we assume that $2 \leq d \leq n(q - 1)$.

\begin{lemma}\label{projafim}
Let $f \in \fq[X_0,\ldots , X_n]_d$ be a polynomial with a nonempty  support 
$S$.  If there exists a hyperplane $H \subset \mathbb{P}^n(\fq)$
such that $S \cap H = \emptyset$ and $|f| > W^{(1)}_{\prm}(n, d)$ then 
$|f| \geq W^{(2)}_{\RM}(n, d - 1)$.
\end{lemma}
\begin{proof}
After a projective transformation we may assume that $H$ is the hyperplane 
defined by $X_0 = 0$.  Writing $f = X_0 f_1 + f_2$,  
where $f_1 \in \fq[X_0, \ldots, X_n]_{d - 1}$ and  $f_2 \in \fq[X_1, \ldots, 
X_n]_d$, from $S \cap H = \emptyset$ we get that $f_2$ vanishes on $H$ and a 
fortiori on $\mathbb{P}^n(\fq)$, so $f_2 \in J_q$. Let $g$ be the 
polynomial obtained from $f_1$ by evaluating $X_0 = 1$, then $g$ is not zero 
(otherwise $S = \emptyset$) and $\deg(g) \leq d - 1$. Considering 
$g$ as a polynomial which evaluates at $\mathbb{A}^n(\fq)$ we see that the 
number $|g|$ of points where $g$ is not zero is equal to $|f|$. Since   
$|f| > W^{(1)}_{\prm}(n, d) = W^{(1)}_{\RM}(n, d - 1)$ we must have $|g| \geq 
W^{(2)}_{\RM}(n, d - 1)$.
\end{proof}

In what follows the integers $k$ and $\ell$ will always be the ones uniquely 
defined 
by the equality   
\[ 
d - 1 = k (q - 1) + \ell
\] 
with $0 \leq k \leq n - 1$ and $0 < \ell \leq q - 1$.
Then, from the data on the minimum distance of $\RM(n,d)$ we get, for $0 \leq r 
\leq n$, that 
\[
W^{(1)}_{\prm}(r, d) = W^{(1)}_{\RM}(r, d - 1) = 
\left\{
\begin{matrix}
1                   & \text{if } 0 \le r \le k ; \\
(q - \ell)q^{r-k-1} &  \text{if } k < r \leq n\, . 
\end{matrix}
\right.
\]

\begin{proposition} \label{desigualdadeS1}
Let $S \subset \mathbb{P}^n(\fq)$ be a nonempty set and
assume that $S$ has the following properties:
\begin{enumerate}
	\item \label{propS1} $|S| <
	\left(
	1 + \dfrac{1}{q}
	\right)
	(q  - \ell) q^{n-k-1} 
	$.
	\item For every linear subspace $G \subset \mathbb{P}^n(\fq)$
	of  dimension $s$, with $s \in \{1, \ldots, n-1\}$,  either $S \cap G = 
	\emptyset$ or
	$|S \cap G| \ge 
	W^{(1)}_{\prm}(s, d)$.
\end{enumerate}
Then there  exists a hyperplane $H \subset \mathbb{P}^n(\fq)$  such that $S 
\cap H = 
\emptyset$.
\end{proposition}

%
%
\begin{proof} We start by noting that 
\[
|S| <	 \left(
	1 + \dfrac{1}{q}
	\right)
	(q  - \ell) q^{n-k-1}  = \left( q - \ell + 1 - \dfrac{\ell}{q} \right) 
	q^{n-k-1}  < q^{n-k} \le q^n < |\mathbb{P}^n(\fq)| \, , 
\]
so let $0 \leq r < n$ be the largest integer such that there is a linear 
subspace $F \subset \mathbb{P}^n(\fq)$ of dimension $r$  satisfying $S \cap F = 
\emptyset$, we want to show that $r = n - 1$. Let 
$$
\mathcal{G}_F = \{ G \text{ a linear subspace of } \mathbb{P}^n(\fq) \mid F 
\subset G \text{\ \ and } \dim G = r +1
\}\, .
$$
The intersection of  two distinct  elements of $\mathcal{G}_F$ is $F$, $|G| 
= (q^{r + 2} - 1)/(q - 1)$ for all $G \in \mathcal{G}_F$ and any point of 
$\mathbb{P}^n(\fq)$ outside $F$ belongs to some $G \in \mathcal{G}_F$
hence 
\[
|\mathcal{G}_F|\left(\dfrac{q^{r + 2} - 1}{q - 1} - \dfrac{q^{r + 1} - 1}{q - 
1}\right) + 
\dfrac{q^{r + 1} - 1}{q - 1} = \dfrac{q^{n + 1} - 1}{q - 1} 
\]
and we get $|\mathcal{G}_F| = (q^{n - r} - 1)/(q - 1)$. From 
$S = \bigcup_{G \in \mathcal{G}_F} (S \cap G)$, $S \cap F = \emptyset$ and 
property 2 we get  
\[
|S| =  \sum_{G \in \mathcal{G}_F} |S \cap G| \ge \dfrac{q^{n - r} - 1}{q - 1} 
\cdot 
W^{(1)}_{\prm}(r+1, d)  \, .
\]
Assume that $r < k$, then $W^{(1)}_{\prm}(r + 1, d) = 1$ and from property 1 we 
get $(q^{n-r}-1)/(q-1)  < (1 + 1/q) (q  - \ell) q^{n-k-1}$. Since the left-hand 
side decreases with $r$ we plug in $r = k - 1$ and get   
\[
\dfrac{q^{n-k+1}-1}{q-1} = q^{n - k} + \cdots + q + 1   <  \left(
q - \ell + 1 - \dfrac{\ell}{q}
\right)
 q^{n-k-1} \leq q^{n - k} - q^{n - k - 2} 
\]
which is absurd. Now we assume that $k \leq r \leq n -1$, and again from 
property 1 we get  	
\[
\dfrac{q^{n-r}-1}{q-1} (q - \ell) q^{r-k} <  \left(
1 + \dfrac{1}{q}
\right)
(q  - \ell) q^{n-k-1},
\]
hence $q^{n-r}-1   <
q^{n-r}   - q^{n-r-2}$ which is only possible when $r = n -1$.  
\end{proof}

\begin{proposition} \label{desigualdadeS2}
Let $S \subset \mathbb{P}^n(\fq)$ be a nonempty set and
assume that $S$ has the following properties:
\begin{enumerate}
\item  $|S| \le 
\left(1 + \dfrac{1}{(q-\ell)}  \right)
(q - \ell)q^{n-k-1}=(q - \ell + 1)q^{n-k-1}$.
\item For every linear subspace $G \subset \mathbb{P}^n(\fq)$
of  dimension $s$, with $s \in \{1, \ldots, n-1\}$,  either $S \cap G = 
\emptyset$ or
$|S \cap G| \ge 
W^{(1)}_{\prm}(s, d)$.
\end{enumerate}
Then there exists  $r \ge k$ and a linear subspace 
$H_r \subset \mathbb{P}^n(\fq)$ of dimension  $r$ such that $S \cap H_r = 
\emptyset$.
\end{proposition}
\begin{proof}
We start as in the proof of the previous Lemma, and observe that 
\[
|S| \le  (q-\ell +1)q^{n-k-1} \le q^{n-k} \le q^n < |\mathbb{P}^n(\fq)| \, ,
\]
so let $0 \leq r < n$ be the largest integer such that there is a linear 
subspace $F \subset \mathbb{P}^n(\fq)$ of dimension $r$  satisfying $S \cap F = 
\emptyset$. Let 
$$
\mathcal{G}_F = \{ G \text{ a linear subspace of } \mathbb{P}^n(\fq) \mid F 
\subset G \text{\ \ and } \dim G = r +1
\}\, , 
$$
as before we have $|\mathcal{G}_F| = (q^{n - r} - 1)/(q - 1)$ and
\[
|S| =  \sum_{G \in \mathcal{G}_F} |S \cap G| \ge \dfrac{q^{n - r} - 1}{q - 1} 
\cdot 
W^{(1)}_{\prm}(r+1, d)  \, .
\]
Assume that $r < k$, then $W^{(1)}_{\prm}(r+1, d) = 1$ and from property 1 we 
get  
\[
\dfrac{q^{n-r}-1}{q-1} = q^{n-r-1}+\cdots + 1 \leq   (q - \ell + 1)q^{n-k-1}  = 
q^{n-k} - 
(\ell - 1 )q^{n-k-1} \le q^{n-r-1}
\]
which is absurd, so we must have $k \leq r \leq n - 1$.
\end{proof}

%

%
%

\section{Main results}

In this section we determine the next-to-minimal weight for the binary 
projective Reed-Muller codes. Recall that we are 
assuming that $2 \leq d \leq 
n(q - 1)$ so if $q = 2$ we have $n \geq 2$. Also, from $d - 1 = k (q - 1) + 
\ell$, with $0 \leq k \leq n -1$ and $0 < \ell \leq q - 1$ we see that when $q 
= 2$ we have $\ell = 1$ and $d - 1 = k + 1$, so from $2 \leq d \leq n$ we get 
$0 \leq 
k \leq n - 2$.
We recall that from \cite[Theorem 
9]{ballet} we have  
\[
W^{(2)}_{\RM}(n, d - 1) = \left\{ 
\begin{array}{ll} 
2^n & \text{ if \ }  k = 0  \\
3 \cdot 2^{n-k-2} & \text{ if \ } 0 < k < n - 2 \\
4 & \text{ if \ } k = n - 2 . \\
\end{array}
\right.
\]

\begin{theorem}  \label{main}
Let $q = 2$ and write $d - 1 = k + 1$. If  $0 < k < n - 2$ or $k = n - 2 \geq 
0$, 
then 
$$W^{(2)}_{\prm}(n, d) = W^{(2)}_{\RM}(n, d - 1),$$
and if $k = 0$ and $n \geq 3$ then 
\[
W^{(2)}_{\prm}(n, d) = 3 \cdot 2^{n - 2} < W^{(2)}_{\RM}(n, d - 1).
\]

\end{theorem}

%
%
\begin{proof}
We start with the case where $k = 0$ and $n \geq 3$. Let 
$f\in \fq[X_0,\ldots , X_n]_d$ be a nonzero polynomial 
such that 
\[
0 \neq |f| < \left( 1 + \dfrac{1}{2} \right) W^{(1)}_{\prm}(n, d)
\]
and observe that 
\[
\left(1 + \dfrac{1}{2}\right) W^{(1)}_{\prm}(n, d) =  3 \cdot 2^{n - 2} < 
W^{(2)}_{\RM}(n, d-1) = 2^n.
\]
Let $S$ be the support of $f$, then $S$ has property 1 of 
Proposition \ref{desigualdadeS1} and  from Lemma \ref{suporte} it also has 
property 2 so 
there exists a hyperplane $H \subset \mathbb{P}^n(\fq)$ such that $S \cap H = 
\emptyset$ and from Lemma  \ref{projafim} we must have  
$|f|= W^{(1)}_{\prm}(n, d)$. This shows that  
\[
 3 \cdot 2^{n - 2} \leq 
W^{(2)}_{\prm}(n, d),
\]
and let $g = X_0 X_3 + X_1 X_2$.
Let $\widetilde{H}_{n-2}$ be the linear subspace defined by $X_0 = X_1 = 0$ and 
note that the support of $g$ does not meet $\widetilde{H}_{n-2}$. Counting the 
number of hyperplanes that contain $\widetilde{H}_{n-2}$ (as in the proof of 
Proposition 
\ref{desigualdadeS1}) we get that a total of three hyperplanes, which we call 
$G_0$, $G_1$ and $G_2$, whose  equations are, respectively,  $X_0 = 0$, $X_1 = 
0$ and $X_0 + X_1 = 0$. Now it is easy to check that in each of these 
hyperplanes 
we have $2^{n - 2}$ points in the support of $g$, hence $|g| = 3 \cdot 2^{n - 
2}$, which settles this case.

Assume now that  $k$ is in the range $0 < k < n-2$, and for $d = k + 2$ let 
$f\in \fq[X_0,\ldots , X_n]_d$ be a nonzero polynomial
such that 
\[
0 \neq |f| < \left( 1 + \dfrac{1}{2} \right) 2^{n-k-1}
\]
(we observe that such a polynomial exists because $W^{(1)}_{\prm}(n, d) = 
2^{n-k-1}$). Let $S$ be the support of $f$, then $S$ has property 1 of 
Proposition \ref{desigualdadeS1} and  from Lemma \ref{suporte} it also has 
property 2 so 
there exists a hyperplane $H \subset \mathbb{P}^n(\fq)$ such that $S \cap H = 
\emptyset$. From Lemma  \ref{projafim} we have that either 
$|f|= W^{(1)}_{\prm}(n, d)$ or $|f| \ge W^{(2)}_{\RM}(n, d - 1)$. Note 
that $W^{(2)}_{\RM}(n, d - 1) = 3 \cdot 2^{n-k-2} = (1 + 1/2) 2^{n-k-1}$, 
so we may conclude that $W^{(2)}_{\prm}(n, d) \geq W^{(2)}_{\RM}(n, d - 1)$, 
and we have already remarked in (\ref{2.1}) that $W^{(2)}_{\prm}(n, d) \leq 
W^{(2)}_{\RM}(n, d - 1)$.

Now we assume that $k = n - 2 \geq 0$, and let $f \in  \fq[X_0,\ldots , X_n]_d$ 
be 
a nonzero polynomial such that $0 \neq | f | \leq 4$ (such polynomial exists 
because $W^{(1)}_{\prm}(n, d) = 2$ in this case, where $d = n$). Then we may 
apply
Proposition \ref{desigualdadeS2} and find that  there exists
a linear subspace $H_r$ of dimension 
 $r \geq n - 2$ 
 such that  $H_r \cap S = \emptyset$. 

If $r=n-1$ then lemma \ref{projafim} implies that
$|f|= W^{(1)}_{\prm}(n, d)$ or $|f| \ge W^{(2)}_{\RM}(n, d - 1)$, and were 
done  because from (\ref{2.1}) we know that
$W^{(2)}_{\prm}(n, d) \leq W^{(2)}_{\RM}(n, d - 1)$. Thus we assume that $r = n 
- 2$ and after a projective transformation, if necessary, we assume that 
$H_{n-2}$ is the linear subspace defined by $X_0 = X_1 = 0$. 
As above we have three hyperplanes which contain $H_{n - 2}$, which 
we call 
$G_0$, $G_1$, $G_2$, and whose  equations are, respectively,  $X_0 = 0$, $X_1 = 
0$ and $X_0 + X_1 = 0$.
 Since $|S \cap G_i| \geq 1$ for $i = 0, 1 ,2$ we get 
that 
$| f | \geq 3 > W^{(1)}_{\prm}(n, d) = 2$ (recall that $d = n$). Thus $3 \leq | 
f| 
\leq 4$ and we want to prove that $| f | = 4 = W^{(2)}_{\RM}(n, d - 1)$. 

Let's 
assume, by means of absurd, that $| f | = 3$, so that $|S \cap G_i| = 1$ for $i 
= 0, 1, 2$ and let $(0:1:Q_0)$, $(1:0:Q_1)$ and 
$(1:1:Q_2)$ be the points of intersection of $S$ with $G_0$, $G_1$ and $G_2$ 
respectively. Write $f$ as $f = X_0(X_1 f_2 + f_1 )  + X_1 f_0 + f_3$, 
with   
$f_2 \in \fq[X_0,\ldots , X_n]_{d - 2}$, $f_1 \in \fq[X_0,X_2, \ldots, X_n]_{d 
- 1}$, $f_0 \in \fq[X_1,\ldots , X_n]_{d - 1}$ and $f_3 \in \fq[X_2,\ldots , 
X_n]_{d}$.  
Since $f$ vanishes on $H_{n - 2}$ we get that $f_3$ vanishes on $H_{n - 2}$ 
and a fortiori on $\mathbb{P}^n(\fq)$, so we may assume that $f = X_0(X_1 f_2 + 
f_1 )  + X_1 f_0$. From the definitions of $Q_0$, $Q_1$ and $Q_2$  
 we get that $f_0(1, Q_0) =  
1$,  $f_1(1,Q_1) = 1$ and 
\[
g(Q_2) := f_2(1, 1, Q_2) + f_1(1, Q_2) + f_0(1, Q_2) = 1.
\]
Observe that  $f_2(1,1,X_2,\ldots , X_n)$ is the zero polynomial or a 
polynomial of 
degree $d - 2 = n - 2$ taking values on $\mathbb{A}^{n-1}(\fq)$, and in the 
latter case we have that either 
\[
|f_2(1,1,X_2,\ldots , X_n)| = W^{(1)}_{\RM}(n - 1, n - 2) = 2  
\]
or
\begin{equation} \label{4}
|f_2(1,1,X_2,\ldots , X_n)| \geq W^{(2)}_{\RM}(n - 1, n - 2) = 4.
\end{equation}
We cannot have $f_2(1,1,X_2,\ldots , X_n) = 0$ otherwise, if $Q_0 \neq Q_1$ we 
would have $g(Q_0) = f_0(1, Q_0) = 1$ and $g(Q_1) = f_1(1, Q_1) = 1$, and 
if $Q_0 = Q_1$ then $g(Q) = 0$ for all $Q \in \mathbb{A}^{n-1}(\fq)$, so in 
both cases we have a contradiction with $|S \cap G_2| = 1$. So we assume that 
$f_2(1,1,X_2,\ldots , X_n) \neq 0$ and in what follows we show that also in 
this case we cannot have $|S \cap 
G_2| = 1$, which will conclude the proof that $|f| = 4$. We split the proof in 
two parts. \\
I)Suppose that $f_2(1,1,Q_2) = 0$, from $g(Q_2) = 1$ we must have $Q_2 = Q_0$ 
or 
$Q_2 = Q_1$, and $Q_0 \neq Q_1$. We will assume that $Q_2 = Q_0$, the case 
where $Q_2 = Q_1$ being similar. We know that there are at least two distinct 
points $Q_3, Q_4 \in \mathbb{A}^{n-1}(\fq)$ such that 
$f_2(1,1,Q_i) = 1$ for $i = 3, 4$. Clearly   $Q_3 \neq Q_2$, if $Q_3 \neq Q_1$ 
we get 
$g(Q_3)  = 1$ and if $Q_3 = Q_1$ then $Q_4 \neq Q_1$, $Q_4 \neq Q_2 = 
Q_0$ and $g(Q_4) = 1$, so  $|S \cap G_2| >  1$. \\
II) Now we assume that $f_2(1,1,Q_2) = 1$, from $g(Q_2) = 1$ we get that either 
$Q_2 = Q_0 = Q_1$,  or 
$Q_2 \neq Q_0$ and $Q_2 \neq Q_1$. From $|f_2(1,1,X_2,\ldots , X_n)|  \geq 2$ 
we know that there exists $Q_3 \neq Q_2$ such that $f_2(1,1,Q_3) = 1$, so  if 
$Q_2 = Q_0 = Q_1$ then $g(Q_3) = 1$ and $|S \cap G_2| >  1$. Thus we assume 
now that $Q_2 \neq Q_0$ and $Q_2 \neq Q_1$.  If $Q_0 = Q_1$ then in both cases  
$Q_3 = Q_0 = Q_1$ or $Q_3 \neq Q_0 = Q_1$  we have $g(Q_3) = 1$. So now we 
consider the case where $Q_0 \neq Q_1$. If $Q_3 \notin \{Q_0, Q_1\}$ then 
$g(Q_3) = 1$, if $Q_3 \in \{Q_0, Q_1\}$, say $Q_3 = Q_0$ (the case where $Q_3 = 
Q_1$  is similar) then $g(Q_1) = 1$ when $f_2(1,1,Q_1) = 0$,   if 
$f_2(1,1,Q_1) = 1$ then we already have three distinct points of $\mathbb{A}^{n 
- 1}(\fq)$ which are not zeros of $f_2(1,1,X_2,\ldots , X_n)$ (namely, $Q_1$, 
$Q_2$ and $Q_3$) so from inequality \eqref{4} above there is a point 
$\widetilde{Q}$, 
distinct from $Q_1$, $Q_2$ and $Q_3 = Q_0$ such that $f_2(1,1,\widetilde{Q}) = 
1$,  hence  $g(\widetilde{Q}) = 1$ and again $|S \cap G_2| >  1$. This 
completes the 
proof of the 
case $k = n - 2$ and of the Theorem. 
\end{proof}

In \cite{delsarte-et-al}  Delsarte et al.\ proved that the codewords of minimal 
weight in $\RM(n, d)$ are such that their support is the union of certain 
affine subspaces of $\mathbb{F}_q^{q^n}$, or equivalently, that these codewords 
may be obtained as the 
evaluation of polynomials whose classes in $\fq[X_1, \ldots, X_n]/I_q$ may be 
represented by the product of $d$ polynomials of degree 1.  In \cite{leduc} the 
author proves a similar result for the next-to-minimal codewords of $\RM(n, 
d)$, in the case where $q \geq 3$. In \cite{rolland2} (see also \cite{ballet}) 
the author proves that also for $\prm(n,d)$ the minimal weight codewords may be 
characterized as being the evaluation of certain homogeneous polynomials whose 
classes in  $\fq[X_0, \ldots, X_n]/J_q$ can be written as  
the product of linear factors, so that the zeros of such polynomials are over a 
union of hyperplanes. As a byproduct of the above proof we see that, for $q = 
2$, such statement is not true for the support of the next-to-minimal 
codewords, since in the case where $k = 0$ and $n \geq 3$ we presented a 
codeword whose zeros form an irreducible quadric in $\mathbb{P}^n(\fq)$.

\bibliographystyle{plain}

\end{document}